%% file: shda.tex
\documentclass[12pt,a4paper]{article}

\usepackage{titlesec}
\titleformat*{\subsection}{\normalsize\bfseries}
\usepackage{tocloft}
\usepackage{authblk}
\usepackage{amsmath,amssymb,amscd,amsthm,latexsym,accents,stmaryrd}
\usepackage{mathrsfs}
\usepackage{mathtools}
\usepackage{amsfonts}
\usepackage{bbm}
\usepackage[latin9]{inputenc}
\usepackage[all,cmtip]{xy}
\usepackage{enumerate}
\usepackage{url}
\usepackage{hyperref}
\usepackage{tikz}
\usetikzlibrary{cd}
\usetikzlibrary{decorations.text}
\usetikzlibrary{patterns}

\usetikzlibrary{automata,positioning}

\makeatletter
\def\thanks#1{\protected@xdef\@thanks{\@thanks
		\protect\footnotetext{#1}}}
\makeatother

\makeatletter
\newcommand{\subjclass}[2][2020]{%
	\let\@oldtitle\@title%
	\gdef\@title{\@oldtitle\footnotetext{#1 \emph{Mathematics subject classification:} #2}}%
}
\newcommand{\keywords}[1]{%
	\let\@@oldtitle\@title%
	\gdef\@title{\@@oldtitle\footnotetext{\emph{Key words and phrases:} #1.}}%
}
\makeatother

\makeatletter
\DeclareFontEncoding{LS2}{}{\@noaccents}
\makeatother
\DeclareFontSubstitution{LS2}{stix}{m}{n}
\DeclareSymbolFont{largesymbolsstix}{LS2}{stixex}{m}{n}
\DeclareMathDelimiter{\lbrbrak}{\mathopen}{largesymbolsstix}{"EE}{largesymbolsstix}{"14}
\DeclareMathDelimiter{\rbrbrak}{\mathclose}{largesymbolsstix}{"EF}{largesymbolsstix}{"15}

\makeatletter
\@namedef{subjclassname@2020}{%
	\textup{2020} Mathematics Subject Classification}
\makeatother

\newcommand{\nocontentsline}[3]{}
\newcommand{\tocless}[2]{\bgroup\let\addcontentsline=\nocontentsline#1{#2}\egroup}

\theoremstyle{definition}
\newtheorem{defin}{Definition}[section]
\newtheorem{rem}[defin]{Remark}

\theoremstyle{plain}
\newtheorem{theor}[defin]{Theorem}
\newtheorem{lem}[defin]{Lemma}
\newtheorem{prop}[defin]{Proposition}
\newtheorem{cor}[defin]{Corollary}
\newtheoremstyle{dotless-thm}
{3pt}
{3pt}
{}
{}
{}
{.}
{.5em}
{}
\theoremstyle{dotless-thm}

\def\Q{{\mathcal{A}}}

\def\Q{{\mathcal{Q}}}

\allowdisplaybreaks

\author{Thomas Kahl\thanks{This research was partially supported by FCT (\emph{Fundação para a Ciência e a Tecnologia}, Portugal) through projects UIDB/00013/2020 and UIDP/00013/2020.} }

\affil{\small{Centro de Matem\'atica,
	Universidade do Minho,\\ Campus de Gualtar,
	4710-057 Braga,
	Portugal\\
\texttt{kahl@math.uminho.pt}}
}

\begin{document}

\title{On symmetric higher-dimensional automata and bisimilarity}

\date{}



\maketitle 

\begin{abstract}
It is shown that a higher-dimensional automaton is hhp-bisimilar to the free symmetric HDA generated by it. Consequently, up to hereditary history-preserving bisimilarity, ordinary HDAs and symmetric HDAs are models of concurrency with the same expressive power. 
\end{abstract}

\section*{Introduction}

A higher-dimensional automaton (HDA) is a precubical set with an initial state, a set of final states, and a labeling on 1-cubes such that opposite edges of 2-cubes have the same label \cite{vanGlabbeek, Pratt}. An HDA is thus
a labeled transition system (or an ordinary automaton) with two- and higher-dimensional cubes linking its states and transitions. An \(n\)-cube in an HDA indicates that the \(n\) transitions starting at its origin are independent in the sense that they may occur in any order, or even simultaneously,  without any observable difference. It has been shown in \cite{vanGlabbeek} that HDAs are a very expressive model of concurrency. 

An important category of HDAs is the one of symmetric HDAs, i.e., HDAs with symmetric underlying precubical sets. The construction of HDAs from other models of concurrency often yields symmetric HDAs (see, e.g., \cite{GaucherCombinatorics, vanGlabbeek, GoubaultLabCubATS, GoubaultMimram}). Ordinary and symmetric HDAs are related by an adjunction: a symmetric HDA is, in particular, an HDA, and conversely, every HDA freely generates a symmetric HDA. In this paper, we compare ordinary and symmetric HDAs with respect to hereditary history-preserving bisimilarity in the sense of van Glabbeek \cite{vanGlabbeek} and establish that an HDA and the free symmetric HDA generated by it are hhp-bisimilar. This result implies as a consequence that, up to hereditary history-preserving bisimilarity, ordinary and symmetric HDAs are equally  expressive models of concurrency.

\renewcommand\contentsname{\normalsize Contents}
\renewcommand{\cftsecleader}{\cftdotfill{\cftdotsep}}
\setlength{\cftparskip}{-10pt}
\setlength{\cftbeforetoctitleskip}{0.8cm}
\setlength{\cftaftertoctitleskip}{0.7cm}
\renewcommand{\cftsecfont}{\normalfont}
\renewcommand{\cftsecpagefont}{\normalfont}
\tableofcontents

\section{Precubical sets and HDAs} \label{SecPrel}

This section recalls the definitions of precubical set and higher-dimensional automaton.

\subsection*{Precubical sets} \label{precubs}

A \emph{precubical set} is a graded set \(P = (P_n)_{n \geq 0}\) with \emph{face maps} \(d^k_i\colon P_n \to P_{n-1}\) \((n>0,\;k= 0,1,\; i = 1, \dots, n)\) satisfying the relations \(d^k_i d^l_{j}= d^l_{j-1} d^k_i\) \((k,l = 0,1,\; i<j)\). If \(x\in P_n\), we say that \(x\) is of \emph{degree} \(n\). The elements of degree \(n\) are called the \emph{\(n\)-cubes} of \(P\). The elements of degree \(0\) are also called the \emph{vertices} of \(P\), and the \(1\)-cubes are also called the \emph{edges} of \(P\). The \emph{\(i\)th starting edge} of a cube \(x\) of degree \(n > 0\) is the edge \(e_ix = d_1^{0} \cdots d_{i-1}^ {0}d_{i+1}^{0}\cdots d_n^{0}x\). A \emph{morphism} of precubical sets is a morphism of graded sets that is compatible with the face maps. The category of precubical sets can be seen as the presheaf category of functors \(\square^{\textsf{op}} \to {\mathsf{Set}}\) where \(\square\) is the small subcategory of the category of topological spaces whose objects are the standard \(n\)-cubes \([0,1]^n\) \((n \geq 0)\) and whose nonidentity morphisms are composites of the \emph{coface maps} \(\delta^k_i\colon [0,1]^n\to [0,1]^{n+1}\) (\(k \in \{0,1\}\), \(n \geq 0\), \(i \in  \{1, \dots, n+1\}\)) given by \(\delta_i^k(u_1,\dots, u_n)= (u_1,\dots, u_{i-1},k,u_i \dots, u_n)\).

\subsection*{Higher-dimensional automata} \label{HDAdef}

Throughout this paper, let \(\Sigma\) be an alphabet. A \emph{higher-di\-mensional automaton} (HDA) over \(\Sigma\) is a tuple \[\Q = (P_{\Q},I_{\Q},F_{\Q},\lambda_{\Q})\] where  \(P_{\Q}\) is a precubical set, \({I_{\Q} \in (P_{\Q})_0}\) is a vertex, called the \emph{initial state}, \({F_{\Q} \subseteq (P_{\Q})_0}\) is a (possibly empty) set of \emph{final states}, and \(\lambda_\Q \colon (P_{\Q})_1 \to \Sigma\) is a map, called the \emph{labeling function}, such that  \(\lambda_{\Q} (d_i^0x) = \lambda_{\Q} (d_i^1x)\) for all \(x \in (P_{\Q})_2\) and \(i \in \{1,2\}\) \cite{vanGlabbeek}. Higher-dimensional automata form a category, in which a morphism from an HDA \(\Q\) to an HDA \(\Q'\) is a morphism of precubical sets  \(f\colon P_\Q \to P_{Q'}\) such that \(f(I_{\Q}) = I_{\Q'}\), \(f(F_{\Q}) \subseteq F_{\Q'}\), and  \(\lambda_{\Q'}(f(x)) =  \lambda_{\Q}(x)\) for all \(x \in (P_\Q)_1\).

\section{The precubical set of permutations}

It is well known that the family of symmetric groups can be given the structure of a \emph{skew-simplicial} or \emph{crossed simplicial group} \cite{Krasauskas, FiedorowiczLoday}. This implies  that it also can be given the structure of a precubical set. In this section, we describe this structure and prove a number of basic facts about it. Recall that the \emph{symmetric group} \(S_n\) is the set of permutations of \(\{1, \dots, n\}\) with composition as multiplication. Here we understand that \(\{1, \dots , 0\} = \emptyset\) and that \(S_0 = \{id_\emptyset\}\).

\subsection*{The maps \texorpdfstring{\(\downarrow i\)}{} and \texorpdfstring{\(\uparrow i\)}{}}

For an integer \(i\), we define the maps \(\downarrow i\) and \(\uparrow i\) on integers by 
\[m^{\downarrow i} = \left\{\begin{array}{ll}
m, & m \leq i,\\
m-1, & m > i
\end{array} \right. \quad \mbox{and} \quad m^{\uparrow i} = \left\{\begin{array}{ll}
m, & m < i,\\
m + 1, & m \geq i.
\end{array} \right.\]
Note that \(m^{\uparrow i \downarrow i} = m\) and, for \(m \not= i\), \(m^{\downarrow i \uparrow i} = m\). Note also that for \(i <j\),
\[m^{\downarrow j \downarrow i} = m^{\downarrow i \downarrow j-1} = \left\{ \begin{array}{ll}
m, & m \leq i,\\
m-1, & i < m \leq j,\\
m-2, & m > j.	
\end{array}\right.\]
We remark that the maps \(\downarrow i\) and \(\uparrow i\) are used to define the coface and codegeneracy maps in the simplex category, which plays an important role in the theory of simplicial sets (see, e.g., \cite{GoerssJardine}).

\subsection*{The face maps of \texorpdfstring{\(S\)}{}}

For \(n \geq 1\), \(\theta \in S_n\), \(i \in \{1, \dots, n\}\), and \(k \in \{0,1\}\), we define the permutation \(d^k_i\theta \in S_{n-1}\) (using one-line notation) by
\[d^k_i\theta = (\theta(1)^{\downarrow i} \;\; \theta(2)^{\downarrow i} \;\; \cdots \;\; \theta(\theta^{-1}(i)-1)^{\downarrow i}\;\; \theta(\theta^{-1}(i)+1)^{\downarrow i} \;\; \cdots\;\; \theta(n)^{\downarrow i}).\]
Thus, more explicitly,
\begin{align*}
d^k_i\theta(j) &= \left \{ \begin{array}{ll}
\theta(j), & j < \theta^{-1}(i),\; \theta(j) < i,\\
\theta(j) - 1, & j < \theta^{-1}(i),\; \theta(j) > i,\\
\theta(j+1), & j \geq \theta^{-1}(i),\; \theta(j+1) < i,\\
\theta(j+1)-1,& j \geq \theta^{-1}(i),\; \theta(j+1) > i.
\end{array}\right.
\end{align*}
Note that by definition, \(d^0_i\theta = d^1_i\theta\). The face maps \(d^k_i\) turn the graded set \(S\) into a precubical set:

\begin{prop}
	For \(1 \leq i < j \leq n\), \(d^k_id^l_j\theta = d^l_{j-1}d^k_i\theta\).
\end{prop}

\begin{proof}
	Set 
	\[r = \left\{\begin{array}{ll}
	i, & \theta^{-1}(i) < \theta^{-1}(j),\\
	j, & \theta^{-1}(i) > \theta^{-1}(j)
	\end{array}\right. \quad \mbox{and} \quad s = \left\{\begin{array}{ll}
	j , & \theta^{-1}(i) < \theta^{-1}(j),\\
	i, & \theta^{-1}(i) > \theta^{-1}(j).
	\end{array} \right.\]
	Since \(i^{\downarrow j} = i\), we have \(i = \theta(\theta^{-1}(i))^{\downarrow j}\) and therefore  
	\begin{align*}
	d^k_id^l_j\theta &= d^k_i(\theta(1)^{\downarrow j} \;\; \cdots \;\; \theta(\theta^{-1}(j)-1)^{\downarrow j}\;\; \theta(\theta^{-1}(j)+1)^{\downarrow j} \;\; \cdots\;\;  \theta(n)^{\downarrow j})\\
	&=  (\theta(1)^{\downarrow j\downarrow i}  \;\; \cdots \;\;  \theta(\theta^{-1}(r)-1)^{\downarrow j\downarrow i}\;\; \theta(\theta^{-1}(r)+1)^{\downarrow j\downarrow i} \;\; \cdots \;\;\\
	& \;\; \;\; \cdots \;\; \theta(\theta^{-1}(s)-1)^{\downarrow j \downarrow i}\;\; \theta(\theta^{-1}(s)+1)^{\downarrow j \downarrow i} \;\; \cdots\;\;  \theta(n)^{\downarrow j\downarrow i}).
	\end{align*}
	Since \(j^{\downarrow i} = j-1\), we have \(j-1 = \theta(\theta^{-1}(j))^{\downarrow i}\) and therefore 
	\begin{align*}
	d^l_{j-1}d^k_i\theta &= d^l_{j-1}(\theta(1)^{\downarrow i}  \;\; \cdots \;\;  \theta(\theta^{-1}(i)-1)^{\downarrow i}\;\; \theta(\theta^{-1}(i)+1)^{\downarrow i} \;\; \cdots\;\;  \theta(n)^{\downarrow i})\\
	&=  (\theta(1)^{\downarrow i\downarrow j-1}  \;\; \cdots \;\;  \theta(\theta^{-1}(r)-1)^{\downarrow i\downarrow j-1}\;\; \theta(\theta^{-1}(r)+1)^{\downarrow i\downarrow j-1} \;\; \cdots \;\;\\
	& \;\; \;\; \cdots \;\; \theta(\theta^{-1}(s)-1)^{\downarrow i \downarrow j-1}\;\; \theta(\theta^{-1}(s)+1)^{\downarrow i \downarrow j-1} \;\; \cdots\;\;  \theta(n)^{\downarrow i\downarrow j-1}).
	\end{align*}
	Since \(m^{\downarrow j \downarrow i} = m^{\downarrow i \downarrow j-1}\), we have \(d^k_id^l_j\theta = d^l_{j-1}d^k_i\theta\).
\end{proof}

The face maps and the multiplication of \(S\) are compatible in the following sense:

\begin{prop} \label{dsigmatheta}
	Let \(n \geq 1\), \(i \in \{1, \dots, n\}\), and \(k \in \{0,1\}\). Then 
	\begin{enumerate}
		\item \(d^k_iid = id\);
		\item \(d^k_i(\sigma\cdot \theta) = d^k_i\sigma \cdot d^k_{\sigma^{-1}(i)}\theta\) for all \(\sigma, \theta \in S_n\); 
		\item \((d^k_i\theta)^{-1} = d^k_{\theta^{-1}(i)}\theta^{-1}\) for all \(\theta \in S_n\).
	\end{enumerate}
\end{prop}

\begin{proof}
	(1) follows immediately from the definition of \(d^k_iid\).
	
	(2) Both \(d^k_i(\sigma\cdot \theta)\) and  \(d^k_i\sigma \cdot d^k_{\sigma^{-1}(i)}\theta\) are the composite 
	\begin{align*}
	\{1, \dots , n-1\} &\xrightarrow{\uparrow {\theta^{-1}(\sigma^{-1}(i))}}
	\{1, \dots , n\} \setminus \{\theta^{-1}(\sigma^{-1}(i))\}\\
	&\xrightarrow{\theta} \{1, \dots , n\} \setminus \{\sigma^{-1}(i)\}\\
	&\xrightarrow{\downarrow{\sigma^{-1}(i)}}  \{1, \dots , n-1\}\\
	&\xrightarrow{\uparrow{\sigma^{-1}(i)}}  \{1, \dots , n\} \setminus \{\sigma^{-1}(i)\}\\
	&\xrightarrow{\sigma} \{1, \dots , n\} \setminus \{i\}\\ &\xrightarrow{\downarrow i}  \{1, \dots , n-1\}. 
	\end{align*}
	
	(3) By (1) and (2), \(d^k_i\theta \cdot d^k_{\theta^{-1}(i)}\theta^{-1} = d^k_i(\theta\cdot \theta^{-1}) = d^k_iid = id\).
\end{proof}

\subsection*{Permutations and face conditions} 

Proposition \ref{existstheta} below guarantees the existence of a permutation with two compatible predefined faces. The proof requires the following lemma:

\begin{lem} \label{dequal}
	Consider permutations \(\sigma, \theta \in S_n\) \((n\geq 1)\), and let \(i \in \{1, \dots, n\}\) such that \(d^0_i\sigma = d^0_i\theta\) and \(\sigma^{-1}(i) = \theta^{-1}(i)\). Then \(\sigma = \theta\).
\end{lem}

\begin{proof} 
	By definition, 
	\begin{align*}
	d^0_i\sigma &= (\sigma (1)^{\downarrow i} \;\; \cdots \;\;  \sigma(\sigma^{-1}(i)-1)^{\downarrow i}\;\; \sigma(\sigma^{-1}(i)+1)^{\downarrow i} \;\; \cdots\;\;  \sigma(n)^{\downarrow i})
	\end{align*}	
	and 
	\begin{align*}
	d^0_i\theta &= (\theta(1)^{\downarrow i}  \;\; \cdots \;\;  \theta(\theta^{-1}(i)-1)^{\downarrow i}\;\; \theta(\theta^{-1}(i)+1)^{\downarrow i} \;\; \cdots\;\;  \theta(n)^{\downarrow i}).
	\end{align*}
	Since \(d^0_i\sigma = d^0_i\theta\) and \(\sigma^{-1}(i) = \theta^{-1}(i)\), we have \(\sigma (j)^{\downarrow i} = \theta (j)^{\downarrow i}\) for all \(j \not= \sigma^{-1}(i) = \theta^{-1}(i)\). For these \(j\), \(\sigma(j) \not= i \not= \theta(j)\) and therefore \(\sigma(j) = \sigma (j)^{\downarrow i \uparrow i} = \theta(j)^{\downarrow i \uparrow i} = \theta(j)\). Since \(\sigma(\sigma^{-1}(i)) = i = \theta(\theta^{-1}(i))\), we have \(\sigma(j) = \theta(j)\) for all \(j\in \{1, \dots, n\}\).   
\end{proof}

\begin{prop} \label{existstheta}
	Consider permutations \(\alpha, \beta \in S_n\) \((n\geq 1)\), and let \(r \leq s\) be integers such that \(d^0_r\alpha = d^0_s\beta\). Then there exists a permutation \(\theta \in S_{n+1}\) such that \(d^0_r\theta = \beta\) and \(d^0_{s+1}\theta = \alpha\). If \(\alpha^{-1}(r) \leq \beta^{-1}(s)\), \(\theta\) may be chosen such that \(\theta^{-1}(r) < \theta^{-1}(s+1)\). If \(\alpha^{-1}(r) \geq \beta^{-1}(s)\), \(\theta\) may be chosen such that \(\theta^{-1}(r) > \theta^{-1}(s+1)\).   
\end{prop}

\begin{proof}
	(i) Suppose that \(\alpha^{-1}(r) \leq \beta^{-1}(s)\). Set
	\begin{align*}
	\theta &= (\beta(1)^{\uparrow r} \;\;\cdots \;\; \beta(\alpha^{-1}(r)-1)^{\uparrow r} \;\; r \;\; \beta(\alpha^{-1}(r))^{\uparrow r} \;\; \cdots\;\; \beta(n)^{\uparrow r}). 
	\end{align*}
	Then \(d^0_{r}\theta = \beta\). Since \(r < s+1\), 
	\[d^0_rd^0_{s+1}\theta = d^0_sd^0_r\theta =  d^0_s\beta = d^0_r\alpha.\]
	Since \(\alpha^{-1}(r) \leq \beta^{-1}(s)\), we have \(\alpha^{-1}(r) = \theta^{-1}(r) < \theta^{-1}(s+1)\) and therefore  
	\begin{align*}
	(d^0_{s+1}\theta)^{-1}(r) &= d^0_{\theta^{-1}(s+1)}\theta^{-1}(r) = \theta^{-1}(r) = \alpha^{-1}(r). 
	\end{align*}
	By Lemma \ref{dequal}, it follows that \(d^0_{s+1}\theta = \alpha\).
	
	(ii) If \(\alpha^{-1}(r) \geq \beta^{-1}(s)\), a similar argument shows that 
	\begin{align*}
	\theta &= (\alpha(1)^{\uparrow s+1} \;\;\cdots \;\; \alpha(\beta^{-1}(s)-1)^{\uparrow s+1} \;\; s+1 \;\; \alpha(\beta^{-1}(s))^{\uparrow s+1} \;\; \cdots\;\; \alpha(n)^{\uparrow s+1})
	\end{align*}
	has the required properties.
\end{proof}

\subsection*{Cubical identities and permutations} 

The cubical identities \(d^k_id^l_j = d^l_{j-1}d^k_i\) of precubical sets can be generalized using permutations. This is done after the next lemma. 

\begin{lem} \label{keylem}
	Consider a permutation \(\theta \in S_n\) \((n \geq 2)\). Let \(1 \leq i < j \leq n\) and \(k,l \in \{0,1\}\). If \(d^l_{\theta(j)}\theta(i) < \theta(j)\), then \(d^l_{\theta(j)}\theta(i) = \theta(i)\) and \(d^k_{\theta(i)}\theta(j-1) = \theta(j) -1\). Else \(d^l_{\theta(j)}\theta(i) = \theta(i) -1\) and \(d^k_{\theta(i)}\theta(j-1) = \theta(j)\).			
\end{lem}

\begin{proof}
	We will suppose that \(d^l_{\theta(j)}\theta(i) < \theta(j)\). The other case is analogous. Since \(i < j = \theta^{-1}(\theta(j))\), we have \(d^l_{\theta(j)}\theta(i) = \theta(i)\) or \(d^l_{\theta(j)}\theta(i) = \theta(i) - 1\). If we had \(d^l_{\theta(j)}\theta(i) = \theta(i) - 1\), we would have \(\theta(i) > \theta(j)\) and thus \(d^l_{\theta(j)}\theta(i) \geq \theta(j)\). Hence \(d^l_{\theta(j)}\theta(i) = \theta(i)\). Since \(j-1 \geq i = \theta^{-1}(\theta(i))\) and \(\theta (j-1 +1) = \theta(j) > d^l_{\theta(j)}\theta(i) = \theta(i)\), we have \(d^k_{\theta(i)}\theta(j-1) = \theta(j) -1\).
\end{proof}

\begin{prop} \label{precubtheta}
	Let \(P\) be a precubical set, and let \(n \geq 2\), \(1 \leq i < j \leq n\), \(k,l \in \{0,1\}\), \(x \in P_n\), and \(\theta \in S_n\). Then
	\begin{enumerate}[(i)]
		\item \(d^k_{d^l_{\theta(j)}\theta (i)}d^l_{\theta(j)}x = d^l_{d^k_{\theta(i)}\theta(j-1)}d^k_{\theta(i)}x\); 
		\item \(d^k_{(d^l_j\theta)^{-1}(i)}d^l_{\theta^{-1}(j)}x = d^l_{(d^k_i\theta)^{-1}(j-1)}d^k_{\theta^{-1}(i)}x\).
	\end{enumerate}
\end{prop}

\begin{proof}
	(i) If \(d^l_{\theta(j)}\theta(i) < \theta(j)\), then, by Lemma \ref{keylem}, \[d^k_{d^l_{\theta(j)}\theta (i)}d^l_{\theta(j)}x = d^l_{\theta(j) -1}d^k_{d^l_{\theta(j)}\theta (i)}x = d^l_{d^k_{\theta(i)}\theta(j-1)}d^k_{\theta(i)}x.\]
	If \(d^l_{\theta(j)}\theta(i) \geq \theta(j)\), then, again by Lemma \ref{keylem}, 
	\[d^l_{d^k_{\theta(i)}\theta(j-1)}d^k_{\theta(i)}x = d^l_{\theta(j)}d^k_{d^l_{\theta(j)}\theta(i)+1}x = d^k_{d^l_{\theta(j)}\theta (i)}d^l_{\theta(j)}x.\]
	
	(ii) By Proposition \ref{dsigmatheta} and (i), \(d^k_{(d^l_j\theta)^{-1}(i)}d^l_{\theta^{-1}(j)}x = d^k_{d^l_{\theta^{-1}(j)}\theta^{-1}(i)}d^l_{\theta^{-1}(j)}x = d^l_{d^k_{\theta^{-1}(i)}\theta^{-1}(j-1)}d^k_{\theta^{-1}(i)}x = d^l_{(d^k_i\theta)^{-1}(j-1)}d^k_{\theta^{-1}(i)}x\).	
\end{proof}

\begin{rem}
	By Proposition \ref{precubtheta}, the graded set \(S\) can be given a second structure of precubical set where the face maps are defined by \(\partial^k_i\theta = d^k_{\theta(i)}\theta\). By Proposition \ref{dsigmatheta}, the map \(\theta \mapsto \theta^{-1}\) is an isomorphism between the two precubical sets of permutations. 
\end{rem}

\section{Symmetric precubical sets and HDAs}

A symmetric HDA is an HDA with symmetric underlying precubical set.
Symmetric precubical sets are usually defined as presheaves on a suitable category of cubes (see, e.g., \cite{GaucherCombinatorics, GoubaultMimram}). Here we define them equivalently as precubical sets with a crossed action of the precubical set \(S\).  We also define free symmetric precubical sets and HDAs, which are central to our work in the following sections.

\subsection*{Crossed actions} A \emph{crossed action} of \(S\) on a precubical set \(P\) is a morphism of graded sets \(S \times P \xrightarrow{} P\), \((\theta, x) \mapsto \theta\cdot x\) satisfying the following three conditions:
\begin{enumerate}
	\item For all \(n \geq 0\) and \(x \in P_n\), \[id\cdot x = x.\]
	\item For all \(n \geq 0\), \(\sigma, \theta \in S_n\), and \(x \in P_n\), \[(\sigma \cdot\theta)\cdot x = \sigma \cdot(\theta \cdot x).\]
	\item For all \(n \geq 1\), \(\theta \in S_n\), \(x \in P_n\), \(i \in \{1, \dots, n\}\), and \(k \in \{0,1\}\),  \[d^k_i(\theta \cdot x) = d^k_i\theta\cdot d^k_{\theta^{-1}(i)}x.\]
\end{enumerate} 
For example, the multiplication \(S \times S \to S\) is a crossed action of \(S\) on itself. 

\subsection*{Symmetric precubical sets} A \emph{symmetric precubical set} is a precubical set \(P\) equipped with a crossed action of \(S\) on \(P\). For example, \(S\) is a symmetric precubical set with respect to the multiplication \(S \times S \to S\). Symmetric precubical sets form a category, in which the morphisms are morphisms of precubical sets that are compatible with the crossed actions. We remark that the category of symmetric precubical sets is isomorphic to the presheaf category \(\mathsf{Set}^{\square_S^{\mathsf{op}}}\) where \(\square_S\) is the subcategory of the category of topological spaces whose objects are the standard \(n\)-cubes \([0,1]^n\) \((n \geq 0)\) and whose morphisms are composites of the coface maps \(\delta^k_i\) defined in Section \ref{SecPrel} and the \emph{permutation maps}  \(t_\theta \colon [0,1]^n \to [0,1]^n\) (\(n \geq 0\), \(\theta \in S_n\)) given by \(t_\theta (u_1, \dots , u_n) = (u_{\theta(1)} \dots, u_{\theta(n)})\). 

\subsection*{Free symmetric precubical sets} 

Let \(P\) be a precubical set. The \emph{free symmetric precubical set} generated by \(P\) is the symmetric precubical set \(SP\) defined by 
\begin{itemize}
	\item \((SP)_n = S_n \times P_n\) \((n \geq 0)\);
	\item \(d^k_i(\theta,x) = (d^k_i\theta, d^k_{\theta^{-1}(i)}x)\) \(({n \geq 1}, {\theta \in S_n}, {x \in P_n}, {1 \leq i \leq n}, {k \in \{0,1\}})\); 
	\item \(\sigma \cdot (\theta, x) = (\sigma \cdot \theta, x)\) \((n \geq 0, \sigma, \theta \in S_n, x \in P_n)\).
\end{itemize}
It follows from Propositions \ref{dsigmatheta} and \ref{precubtheta} that \(SP\) is indeed a symmetric precubical set.
The free symmetric precubical set is functorial. Given a morphism of precubical set \(f \colon P \to Q\), \(Sf \colon SP \to SQ\) is the graded map \(id_S\times f\).

\begin{prop}
	The functor \(S\) from the category of precubical sets to the category of symmetric precubical sets is left adjoint to the forgetful functor.
\end{prop}

\begin{proof}
	Let \(P\) be a precubical set, and let \(Z\) be a symmetric precubical set. The adjunct of a morphism of precubical sets \(f \colon P \to Z\) is the morphism of symmetric precubical sets \(\hat f \colon SP \to Z\) given by \(\hat f(\theta ,x) = \theta\cdot f(x)\). The adjunct of a morphism of symmetric precubical sets \(g \colon SP \to Z\) is the morphism of precubical sets \(\check{g} \colon P \to Z\) given by \(\check{g}(x) = g(id,x)\).
\end{proof}

Recall that the \(i\)th starting edge of a cube \(x\) of degree \(n \geq 1\) of a precubical set is the edge \(e_ix = d_1^{0} \cdots d_{i-1}^ {0}d_{i+1}^{0}\cdots d_n^{0}x\). The starting edges of cubes in a free symmetric precubical set are related as follows to the starting edges of cubes in the generating precubical set:
 
\begin{prop} \label{edgetheta}
	Let \(P\) be a precubical set, and let \((\theta,x) \in (SP)_n\) \((n \geq 1)\). Then for each \(i \in \{1, \dots, n\}\), \(e_i(\theta, x) = (id, e_{\theta^{-1}(i)}x)\).
\end{prop}

\begin{proof}
	We proceed by induction. For \(n = 1\), there is nothing to show. If \(n = 2\), 
	\[e_i(\theta,x) = d^0_{3-i}(\theta,x) = (d^0_{3-i}\theta, d^0_{\theta^{-1}(3-i)}x) = (id, e_{\theta^{-1}(i)}x).\]
	Suppose that \(n > 2\). Consider first the case \(i < n\). By the inductive hypothesis,
	\[e_i(\theta,x) = e_id^0_n(\theta, x) = e_i(d^0_n\theta, d^0_{\theta^{-1}(n)}x) = (id, e_{(d^0_n\theta)^{-1}(i)}d^0_{\theta^{-1}(n)}x).\]
	Since \(i<n\),
	\[(d^0_n\theta)^{-1}(i) = d^0_{\theta^{-1}(n)}\theta^{-1}(i) = \left \{ \begin{array}{ll}
	\theta^{-1}(i), & \theta^{-1}(i) < \theta^{-1}(n),\\
	\theta^{-1}(i)-1, & \theta^{-1}(i) > \theta^{-1}(n).
	\end{array}\right. 
	\]
	Hence 
	\[e_{(d^0_n\theta)^{-1}(i)}d^0_{\theta^{-1}(n)}x = \left \{ \begin{array}{ll}
	e_{\theta^{-1}(i)}d^0_{\theta^{-1}(n)}x = e_{\theta^{-1}(i)}x, & \theta^{-1}(i) < \theta^{-1}(n),\\
	e_{\theta^{-1}(i)-1}d^0_{\theta^{-1}(n)}x = e_{\theta^{-1}(i)}x, & \theta^{-1}(i) > \theta^{-1}(n)
	\end{array}\right.
	\]
	and therefore \(e_i(\theta, x) = (id, e_{\theta^{-1}(i)}x)\).
	
	Suppose now that \(i = n\). By the inductive hypothesis,
	\begin{align*}
		e_n(\theta,x) &= e_{n-1}d^0_{n-1}(\theta, x)\\ &= e_{n-1}(d^0_{n-1}\theta, d^0_{\theta^{-1}(n-1)}x)\\ &= (id, e_{(d^0_{n-1}\theta)^{-1}(n-1)}d^0_{\theta^{-1}(n-1)}x).
	\end{align*}
	We have 
	\begin{align*}
		(d^0_{n-1}\theta)^{-1}(n-1) &= d^0_{\theta^{-1}(n-1)}\theta^{-1}(n-1)\\ &= \left \{ \begin{array}{ll}
		\theta^{-1}(n), & \theta^{-1}(n) < \theta^{-1}(n-1),\\
		\theta^{-1}(n)-1, & \theta^{-1}(n) > \theta^{-1}(n-1).
		\end{array}\right.
	\end{align*}
	Hence 
	\begin{align*}
		e_{(d^0_{n-1}\theta)^{-1}(n-1)}d^0_{\theta^{-1}(n-1)}x &= \left \{ \begin{array}{ll}
		e_{\theta^{-1}(n)}d^0_{\theta^{-1}(n-1)}x, & \theta^{-1}(n) < \theta^{-1}(n-1),\\
		e_{\theta^{-1}(n)-1}d^0_{\theta^{-1}(n-1)}x, & \theta^{-1}(n) > \theta^{-1}(n-1)
		\end{array}\right.\\
		&= e_{\theta^{-1}(n)}x
	\end{align*}
	and therefore \(e_n(\theta, x) = (id, e_{\theta^{-1}(n)}x)\).
\end{proof}

\subsection*{Symmetric HDAs} A \emph{symmetric HDA} is an HDA \(\Q\) equipped with a crossed action of \(S\) on \(P_\Q\). Symmetric HDAs form a category, in which the morphisms are morphisms of HDAs that also are  morphisms of symmetric precubical sets. The \emph{free symmetric HDA} generated by an HDA \(\Q\) is the symmetric HDA \(S\Q\) where \(P_{S\Q} = SP_\Q\), \(I_{S\Q} = (id,I_\Q)\), \(F_{S\Q} = S_0\times F_\Q\), \(\lambda_{S\Q}(id,x) = \lambda_\Q(x)\) \((x \in (P_\Q)_1)\), and the crossed action is the one of \(SP_\Q\). The assignment \(\Q \mapsto S\Q\) defines a functor from the category of HDAs to the category of symmetric HDAs, which is left adjoint to the forgetful functor.

\section{Cube paths} \label{seccubepaths}

Throughout this section, let \(\Q\) denote an HDA. A \emph{cube path} in \(\Q\) is a sequence of cubes and face maps
\[\pi = x_0 \tfrac{d^{k_1}_{i_1}}{} x_1 \tfrac{d^{k_2}_{i_2}}{}  x_2 \tfrac{d^{k_3}_{i_3}}{} \cdots \tfrac{d^{k_m}_{i_m}}{} x_m\]
such that \(x_0 = I_\Q\) and for all \(j \in \{1, \dots, m\}\),  \(2\sum \limits_{p=1}^j k_p \leq j\), \(x_j \in (P_\Q)_{j-2\sum \limits_{p=1}^j k_p}\), 
\(i_j \in \{1, \dots, j-k_j - 2\sum \limits_{p=1}^{j-k_j} k_p\}\), 
and \(d^{k_j}_{i_j}x_{j-k_j} = x_{j-1 + k_j}\) (cf. \cite{vanGlabbeek}). We refer to \(m\) as the \emph{length} of \(\pi\) and  write \(\mathit{end}(\pi)\) for \(x_m\). We write \(\pi \to \pi'\) if \(\pi\) extends to a cube path \(\pi'\). A cube path in \(\Q\) represents a partial execution of the concurrent system modeled by \(\Q\). An example of a cube path is indicated by the thick arrows in the following very simple HDA:
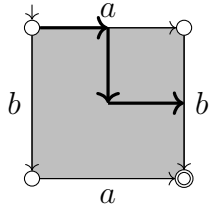
\begin{center}
	\begin{tikzpicture}[initial text={},on grid]  
	\path[draw, fill=lightgray] (0,0)--(2,0)--(2,2)--(0,2)--cycle;
	\node[state,minimum size=0pt,inner sep =2pt,fill=white] (p_0)   {}; 
	\node[state,accepting,minimum size=0pt,inner sep =2pt,fill=white] (p_2) [right=of p_0,xshift=1cm] {};
	\node[state,initial,initial distance=0.2cm,minimum size=0pt,inner sep =2pt,fill=white,initial where=above,initial distance=0.2cm] [above=of p_0, yshift=1cm] (p_3)   {};
	\node[state,minimum size=0pt,inner sep =2pt,fill=white] (p_5) [right=of p_3,xshift=1cm] {}; 
	\path[->] 
	(p_0) edge[below] node {$a$} (p_2)
	(p_3) edge[above]  node {$a$} (p_5)
	(p_3) edge[left]  node {$b$} (p_0)
	(p_5) edge[right]  node {$b$} (p_2);
	\path[->]
	(p_3) edge[above, very thick]  (1,2)
	(1,2) edge[above, very thick]  (1,1)
	(1,1) edge[above, very thick]  (2,1);	
	\end{tikzpicture}
\end{center}
If the 2-cube in this HDA is \(x\), with \(d^0_2x\) the upper horizontal edge  and \(d^1_1x\) the right vertical edge, the depicted cube path is 
\[d^0_1d^0_2x \tfrac{d^0_1}{} d^0_2x \tfrac{d^0_2}{} x \tfrac{d^1_1}{} d^1_1x.\]
The \emph{split trace} of a cube path  
\(\pi = x_0 \tfrac{d^{k_1}_{i_1}}{} x_1 \tfrac{d^{k_2}_{i_2}}{}  x_2 \tfrac{d^{k_3}_{i_3}}{} \cdots \tfrac{d^{k_m}_{i_m}}{} x_m\)
is the sequence \[\mathit{split\mbox{-}trace}(\pi) = ((\lambda_\Q(e_{i_1}x_{1-k_1}),k_1), \dots, (\lambda_\Q(e_{i_m}x_{m-k_m}),k_m))\] 
(see \cite{vanGlabbeek}). The split trace is the sequence of actions starting (with second component \(0\)) and terminating (with second component \(1\)) along the cube path. The split trace of the cube path in the example above is the sequence
\[((a,0),(b,0),(a,1)).\] 
This cube path describes thus an (incomplete) execution of the system modeled by the HDA where first \(a\) starts, then \(b\) starts, and finally \(a\) terminates.

\begin{defin} \cite{vanGlabbeek} \label{adjacent}
	Two cube paths  \[{\pi = x_0 \tfrac{d^{k_1}_{i_1}}{} x_1 \tfrac{d^{k_2}_{i_2}}{}  \cdots \tfrac{d^{k_m}_{i_m}}{} x_m}\quad  \mbox{and} \quad {\gamma = y_0 \tfrac{d^{q_1}_{r_1}}{} y_1 \tfrac{d^{q_2}_{r_2}}{}   \cdots \tfrac{d^{q_m}_{r_m}}{} y_m}\] of the same length \(m\geq 2\) are said to be \emph{\(\ell\)-adjacent} \((1 \leq \ell < m)\), denoted \(\pi \xleftrightarrow{\ell} \gamma\), if \(x_j = y_j\) for all \(j \not= \ell\), \(k_j = q_j\) and \(i_j = r_j\) for all \(j \not= \ell, \ell+1\), and one of the following conditions holds:
	\begin{itemize}
		\item[(i)] \(k_\ell = q_{\ell +1} = k_{\ell +1} = q_\ell = 0\) and \(i_\ell = r_{\ell+1}  < i_{\ell+1} = r_\ell+1\) 
		\item[(ii)] \(q_\ell = k_{\ell +1} = q_{\ell +1} = k_\ell = 0\) and \(r_\ell = i_{\ell+1}  < r_{\ell+1} = i_\ell+1\)
		\item[(iii)] \(k_\ell = q_{\ell +1} = 0\), \(k_{\ell +1} = q_\ell = 1\), and \(i_\ell = r_{\ell+1}  < i_{\ell+1} = r_\ell+1\)
		\item[(iv)] \(q_\ell = k_{\ell +1} = 0\), \(q_{\ell +1} = k_\ell = 1\), and \(r_\ell = i_{\ell+1}  < r_{\ell+1} = i_\ell+1\) 
		\item[(v)] \(k_\ell = q_{\ell +1} = 0\), \(k_{\ell +1} = q_\ell = 1\), and \(i_\ell = r_{\ell+1}+1  > i_{\ell+1} = r_\ell\)
		\item[(vi)] \(q_\ell = k_{\ell +1} = 0\), \(q_{\ell +1} = k_\ell = 1\), and \(r_\ell = i_{\ell+1}+1  > r_{\ell+1} = i_\ell\)
		\item[(vii)] \(k_\ell = q_{\ell +1} = k_{\ell +1} = q_\ell = 1\) and \(i_\ell = r_{\ell+1}+1  > i_{\ell+1} = r_\ell\)
		\item[(viii)] \(q_\ell = k_{\ell +1} = q_{\ell +1} = k_\ell = 1\) and \(r_\ell = i_{\ell+1}+1  > r_{\ell+1} = i_\ell\) 	
	\end{itemize}
\end{defin} 

For example, the cube path considered above and the cube path \[d^0_1d^0_2x \tfrac{d^0_1}{} d^0_2x \tfrac{d^1_1}{} d^1_1d^0_2x \tfrac{d^0_1}{} d^1_1x\]
with split trace \(((a,0),(a,1),(b,0))\) are \(2\)-adjacent, satisfying condition (v).

\subsection*{Cube paths in \texorpdfstring{\(S\Q\)}{} and the map \texorpdfstring{\(\phi\)}{}}

With a cube path \[\pi = (\theta_0,x_0) \tfrac{d^{k_1}_{i_1}}{} (\theta_1,x_1) \tfrac{d^{k_2}_{i_2}}{}  \cdots \tfrac{d^{k_m}_{i_m}}{} (\theta_m,x_m)\] in \(S\Q\), we associate the cube path \[\phi(\pi) =  x_0 \tfrac{d^{k_1}_{\theta_{1-k_1}^{-1}(i_1)}}{} x_1 \tfrac{d^{k_2}_{\theta_{2-k_2}^{-1}(i_2)}}{}  \cdots \tfrac{d^{k_m}_{\theta_{m-k_m}^{-1}(i_m)}}{} x_m\] in \(\Q\). Since \((\theta_0,x_0) = I_{S\Q} = (id, I_\Q)\), we have \(x_0 = I_\Q\). Since for all \(j \in \{1, \dots, m\}\),  \(2\sum \limits_{p=1}^j k_p \leq j\), \((\theta_j,x_j) \in (P_{S\Q})_{j-2\sum \limits_{p=1}^j k_p} = S_{j-2\sum \limits_{p=1}^j k_p}\times (P_\Q)_{j-2\sum \limits_{p=1}^j k_p}\), 
\(i_j\in \{1, \dots, j-k_j - 2\sum \limits_{p=1}^{j-k_j} k_p\}\), and 
\[(d^{k_j}_{i_j}\theta_{j-k_j}, d^{k_j}_{\theta_{j-k_j}^{-1}(i_j)}x_{j-k_j}) =  d^{k_j}_{i_j}(\theta_ {j-k_j},x_{j-k_j}) = (\theta_{j-1+k_j},x_{j-1+k_j}),\] we also have  for all \(j \in \{1, \dots, m\}\), \(x_j \in (P_\Q)_{j-2\sum \limits_{p=1}^j k_p}\), \(\theta_{j-k_j} \in S_{j-k_j - 2\sum \limits_{p=1}^{j-k_j} k_p}\), \(\theta^{-1}_{j-k_j}(i_j) \in \{1, \dots, j-k_j - 2\sum \limits_{p=1}^{j-k_j} k_p\}\), and  \(d^{k_j}_{\theta_{j-k_j}^{-1}(i_j)}x_{j-k_j} = x_{j-1+k_j}\). Hence \(\phi(\pi)\) is indeed a cube path in \(\Q\).

The map \(\phi\) preserves adjacency:

\begin{prop} \label{hhp3}
	Consider cube paths \[\pi = (\theta_0,x_0) \tfrac{d^{k_1}_{i_1}}{} (\theta_1,x_1) \tfrac{d^{k_2}_{i_2}}{}  \cdots \tfrac{d^{k_m}_{i_m}}{} (\theta_m,x_m)\] and \[\gamma = (\sigma_0,y_0) \tfrac{d^{q_1}_{r_1}}{} (\sigma_1,y_1) \tfrac{d^{q_2}_{r_2}}{}  \cdots \tfrac{d^{q_m}_{r_m}}{} (\sigma_m,y_m)\] in \(S\Q\) such that \(\pi \xleftrightarrow[]{\ell} \gamma\). Then \(\phi(\pi) \xleftrightarrow[]{\ell} \phi(\gamma)\).  
\end{prop}

\begin{proof}
	By our hypothesis, \((\theta_j,x_j) = (\sigma_j,y_j)\) for all \(j \not= \ell\), \(k_j = q_j\) and \(i_j = r_j\) for all \(j \not=  \ell, \ell +1\), and one of the conditions of Definition \ref{adjacent} holds. In all cases, \(x_j = y_j\) for all \(j \not= \ell\) and \(k_j = q_j\) and \(\theta_{j-k_j}^{-1}(i_j) = \sigma_{j-q_j}^{-1}(r_j)\) for all \(j \not=  \ell, \ell +1\). We will suppose that condition \ref{adjacent}(i) holds, i.e., 
	\[k_\ell = q_{\ell +1} = k_{\ell +1} = q_\ell = 0\quad \mbox{and}\quad  i_\ell = r_{\ell+1}  < i_{\ell+1} = r_\ell+1.\]
	The arguments in the remaining situations are analogous. We have
	\[(\theta_\ell, x_\ell) = d^0_{i_{\ell+1}}(\theta_{\ell+1},x_{\ell+1}) = (d^0_{i_{\ell+1}}\theta_{\ell+1},d^0_{\theta_{\ell+1}^{-1}(i_{\ell+1})}x_{\ell+1})\]
	and therefore 
	\begin{align*}
	\theta_\ell^{-1}(i_\ell) &= (d^0_{i_{\ell+1}}\theta_{\ell+1})^{-1}(i_\ell) = d^0_{\theta_{\ell+1}^{-1}(i_{\ell+1})}\theta_{\ell+1}^{-1}(i_\ell)\\
	&= \left\{\def\arraystretch{1.4} \begin{array}{ll}
	\theta_{\ell+1}^{-1}(i_\ell)  = \sigma_{\ell+1}^{-1}(r_{\ell + 1}), & \theta_{\ell+1}^{-1}(i_\ell) < \theta_{\ell+1}^{-1}(i_{\ell+1}),\\
	\theta_{\ell+1}^{-1}(i_\ell) - 1 = \sigma_{\ell+1}^{-1}(r_{\ell + 1}) - 1, & \theta_{\ell+1}^{-1}(i_\ell) > \theta_{\ell+1}^{-1}(i_{\ell+1}).
	\end{array}\right.
	\end{align*} 
	Since \(q_{\ell + 1} = 0\), 
	\((\sigma_\ell, y_\ell) = d^0_{r_{\ell+1}}(\sigma_{\ell+1},y_{\ell+1}) = (d^0_{r_{\ell+1}}\sigma_{\ell+1},d^0_{\sigma_{\ell+1}^{-1}(r_{\ell+1})}y_{\ell+1}) = (d^0_{i_{\ell}}\theta_{\ell+1},d^0_{\theta_{\ell+1}^{-1}(i_{\ell})}x_{\ell+1})\)
	and therefore
	\begin{align*}
	\sigma_\ell^{-1}(r_\ell) &= (d^0_{i_{\ell}}\theta_{\ell+1})^{-1}(i_{\ell+1}-1) = d^0_{\theta_{\ell+1}^{-1}(i_{\ell})}\theta_{\ell+1}^{-1}(i_{\ell+1}-1)\\
	&= \left\{ \def\arraystretch{1.4} \begin{array}{ll}
	\theta_{\ell+1}^{-1}(i_{\ell+1}), & \theta_{\ell+1}^{-1}(i_{\ell+1}) < \theta_{\ell+1}^{-1}(i_{\ell}),\\
	\theta_{\ell+1}^{-1}(i_{\ell+1}) - 1, & \theta_{\ell+1}^{-1}(i_{\ell+1}) > \theta_{\ell+1}^{-1}(i_{\ell}).
	\end{array}\right. 
	\end{align*}
	If \(\sigma_{\ell+1}^{-1}(r_{\ell + 1}) = \theta_{\ell+1}^{-1}(i_\ell) < \theta_{\ell+1}^{-1}(i_{\ell+1})\), we obtain
	\[\theta_\ell^{-1}(i_\ell)  =  \sigma_{\ell+1}^{-1}(r_{\ell + 1}) < \theta_{\ell+1}^{-1}(i_{\ell+1}) = \sigma_\ell^{-1}(r_\ell) + 1,\] 
	which shows that \(\phi(\pi)\) and \(\phi(\gamma)\) satisfy condition \ref{adjacent}(i). If \(\sigma_{\ell+1}^{-1}(r_{\ell + 1}) = \theta_{\ell+1}^{-1}(i_\ell) > \theta_{\ell+1}^{-1}(i_{\ell+1})\), we obtain
	\[\sigma_\ell^{-1}(r_\ell) = \theta_{\ell+1}^{-1}(i_{\ell+1}) <   \sigma_{\ell+1}^{-1}(r_{\ell + 1}) = \theta_\ell^{-1}(i_\ell) + 1,\] 
	which shows that \(\phi(\pi)\) and \(\phi(\gamma)\) satisfy condition \ref{adjacent}(ii). Consequently, \(\phi(\pi) \xleftrightarrow[]{\ell} \phi(\gamma)\).
\end{proof}

As a partial converse to Proposition \ref{hhp3}, we have the following result:

\begin{prop} \label{hhp4}
	Consider cube paths \[\pi = (\theta_0,x_0) \tfrac{d^{k_1}_{i_1}}{} (\theta_1,x_1) \tfrac{d^{k_2}_{i_2}}{}  \cdots \tfrac{d^{k_m}_{i_m}}{} (\theta_m,x_m)\] and \[\gamma = (\sigma_0,y_0) \tfrac{d^{q_1}_{r_1}}{} (\sigma_1,y_1) \tfrac{d^{q_2}_{r_2}}{}  \cdots \tfrac{d^{q_m}_{r_m}}{} (\sigma_m,y_m)\] in \(S\Q\) such that \((\theta_j,x_j) = (\sigma_j,y_j)\) for all \(j \not= \ell\), \(k_j = q_j\) and \(i_j = r_j\) for all \(j \not=  \ell, \ell +1\), and \(\phi(\pi) \xleftrightarrow[]{\ell} \phi(\gamma)\). Then \(\pi \xleftrightarrow[]{\ell} \gamma\).
\end{prop}

\begin{proof}
	Since the cube paths \(\phi(\pi)\) and \(\phi(\gamma)\) are \(\ell\)-adjacent, they satisfy one of the conditions of Definition \ref{adjacent}. We will suppose that condition (i) holds. The arguments in the remaining situations are analogous.  
	So our hypothesis is that \(k_\ell = q_{\ell +1} = k_{\ell +1} = q_\ell = 0\)
	and \(\theta_\ell^{-1}(i_\ell) = \sigma_{\ell +1}^{-1}(r_{\ell+1})  < \theta_{\ell + 1 }^{-1}(i_{\ell+1}) = \sigma_{\ell }^{-1}(r_\ell)+1\). We have
	\begin{align*}
	r_\ell &= \sigma_\ell(\sigma_\ell^{-1}(r_\ell)) = d^0_{r_{\ell +1}}\theta_{\ell +1}(\theta_{\ell +1}^{-1}(i_{\ell +1}) -1)\\ &= \left\{\def\setarraystretch{1.4} \begin{array}{ll}
	\theta_{\ell +1}(\theta_{\ell +1}^{-1}(i_{\ell +1})) = i_{\ell +1}, & i_{\ell +1} < r_{\ell +1},\\
	i_{\ell +1} -1, & i_{\ell +1} > r_{\ell +1}
	\end{array} \right.
	\end{align*}
	and 
	\begin{align*}
	i_\ell &= \theta_\ell(\theta_{\ell}^{-1}(i_\ell)) = d^0_{i_{\ell +1}}\theta_{\ell +1}(\theta_{\ell +1}^{-1}(r_{\ell +1}))\\ &= \left\{\def\setarraystretch{1.4} \begin{array}{ll} 
	\theta_{\ell +1}(\theta_{\ell +1}^{-1}(r_{\ell +1})) = r_{\ell +1}, & r_{\ell +1} < i_{\ell +1},\\
	r_{\ell +1} -1, & r_{\ell +1} > i_{\ell +1}.
	\end{array}\right.
	\end{align*}
	Hence either \(r_\ell = i_{\ell +1} < r_{\ell +1} = i_\ell +1\) or \(i_\ell = r_{\ell +1} < i_{\ell +1} = r_\ell +1\), which shows that \(\pi\) and \(\gamma\) satisfy either condition \ref{adjacent}(ii) or condition \ref{adjacent}(i). Thus \(\pi \xleftrightarrow{\ell} \gamma\).
\end{proof}

The last result of this section is the following adjacency lifting property of \(\phi\):

\begin{prop} \label{hhp4lem}
	Let \(\pi = (\theta_0,x_0) \tfrac{d^{k_1}_{i_1}}{} (\theta_1,x_1) \tfrac{d^{k_2}_{i_2}}{}  \cdots \tfrac{d^{k_m}_{i_m}}{} (\theta_m,x_m)\) be a cube path in \(S\Q\), and let \({\rho = y_0 \tfrac{d^{q_1}_{r_1}}{} y_1 \tfrac{d^{q_2}_{r_2}}{}  \cdots \tfrac{d^{q_m}_{r_m}}{} y_m}\) be a cube path in \(\Q\) such that \(\phi(\pi) \xleftrightarrow{\ell} \rho\). Then there exists a cube path \(\gamma\) in \(S\Q\) such that \(\phi(\gamma) = \rho\) and \(\pi \xleftrightarrow[]{\ell} \gamma\). 
\end{prop}

\begin{sloppypar}
\begin{proof}
		Since \(\phi(\pi) \xleftrightarrow{\ell} \rho\), we have \(x_j = y_j\) for all \(j \not= \ell\) and \(k_j = q_j\) for all \(j \not= \ell, \ell+1\). By Proposition \ref{hhp4}, it is therefore enough to construct a cube path \[\gamma = (\sigma_0,y_0) \tfrac{d^{q_1}_{s_1}}{} (\sigma_1,y_1) \tfrac{d^{q_2}_{s_2}}{} \cdots \tfrac{d^{q_m}_{s_m}}{} (\sigma_m,y_m)\] such that \(\sigma_j = \theta_j\) for all \(j \not= \ell\), \(s_j = i_j\) for all \(j \not= \ell, \ell +1\), and \(\phi(\gamma) = \rho\). By our hypothesis,  \(\phi(\pi)\) and \(\rho\) satisfy one of the conditions of Definition \ref{adjacent}. In each of these situations, we start the construction of \(\gamma\) by  setting \(\sigma_j = \theta_j\) for \(j \not= \ell\) and \(s_j = i_j\) for \(j \not= \ell, \ell +1\), as required. Since \(\pi\) is a cube path, \((\sigma_0,y_0) = (\theta_0,x_0) = I_{S\Q}\). Since \(\rho\) is a cube path, \(2\sum \limits_{p=1}^{j}q_j \leq j\) for all \(j\). Since \({\phi(\pi) \xleftrightarrow[]{\ell} \rho}\), we have \(\sum \limits_{p=1}^jk_p = \sum \limits_{p=1}^jq_p\) for \(j \not= \ell\) and \(\theta_{j-k_j}^{-1}(i_j) = r_j\) for \(j \not= \ell, \ell+1\).  Therefore  \((\sigma_j, y_j) = (\theta_j, x_j) \in (P_{S\Q})_{j - 2\sum \limits_{p=1}^jk_p} = (P_{S\Q})_{j - 2\sum \limits_{p=1}^jq_p}\) for \(j \not= \ell\) and 	\(s_j = i_j \in {\{1, \dots, j - k_j - 2\sum \limits_{p=1}^{j-k_j}k_p\}} = \{1, \dots, j - q_j - 2\sum \limits_{p=1}^{j-q_j}q_p\}\),   \(d^{q_j}_{s_j}(\sigma_{j-q_j},y_{j-q_j}) = d^{k_j}_{i_j}(\theta_{j-k_j},x_{j-k_j}) = (\theta_{j-1 + k_j},x_{j-1 + k_j}) = (\sigma_{j-1 + q_j},y_{j-1 + q_j})\), and \(\sigma_{j-q_j}^{-1}(s_j) = \theta_{j-k_j}^{-1}(i_j) = r_j\) for \(j \not=\ell, \ell +1\). To finish the construction of \(\gamma\), it remains to  define the permutation \(\sigma_\ell \in S_{\ell - 2\sum \limits_{p=1}^\ell q_p}\) and to check that with \(s_\ell = \sigma_{\ell-q_\ell}(r_\ell)\) and \({s_{\ell +1} = \sigma_{\ell+1-q_{\ell+1}}(r_{\ell+1})}\), one has  \(d^{q_\ell}_{s_\ell}(\sigma_{\ell-q_\ell},y_{\ell-q_\ell}) = (\sigma_{\ell-1+q_\ell}, y_{\ell-1+q_\ell})\) and  \(d^{q_{\ell+1}}_{s_{\ell+1}}(\sigma_{\ell+1-q_{\ell+1}},y_{\ell +1-q_{\ell+1}}) = (\sigma_{\ell+q_{\ell+1}}, y_{\ell+q_{\ell+1}})\). We will only consider conditions \ref{adjacent}(i), (iii), and (iv). In each of the remaining situations, the arguments are analogous to those used in one of these three cases.

	Suppose first that \(\phi(\pi)\) and \(\rho\) satisfy condition \ref{adjacent}(i). Then \(k_\ell = q_{\ell +1} = k_{\ell +1} = q_\ell = 0\) and \(\theta_{\ell}^{-1}(i_\ell) = r_{\ell+1}  < \theta_{\ell + 1}^{-1}(i_{\ell+1}) = r_\ell+1\). In this situation, we set \(\sigma_\ell = d^0_{\theta_{\ell +1}(r_{\ell +1})}\theta_{\ell +1} = d^0_{\sigma_{\ell +1}(r_{\ell +1})}\sigma_{\ell +1}\). Since \(\sigma_{\ell +1} \in S_{\ell +1 - 2\sum \limits_{p=1}^{\ell +1}q_p}\) and \({r_{\ell +1}\in \{1, \dots, \ell +1 - 2\sum \limits_{p=1}^{\ell +1}q_p\}}\), \(\sigma_\ell\) is a well-defined element of \(S_{\ell - 2\sum \limits_{p=1}^{\ell}q_p}\). We compute 
	\(d^0_{s_{\ell +1}}(\sigma_{\ell +1},y_{\ell +1}) = d^0_{\sigma_{\ell +1}(r_{\ell +1})}(\sigma_{\ell +1},y_{\ell +1}) = (d^0_{\sigma_{\ell +1}(r_{\ell +1})}\sigma_{\ell +1},d^0_{r_{\ell +1}}y_{\ell +1}) = (\sigma_\ell, y_\ell)\)
	and, using Proposition \ref{precubtheta}, 
	\begin{align*}
	d^0_{s_\ell}(\sigma_\ell,y_\ell) &= (d^0_{s_\ell}\sigma_\ell,d^0_{\sigma_\ell^{-1}(s_\ell)}y_\ell)= (d^0_{\sigma_\ell(r_\ell)}\sigma_\ell,d^0_{\sigma_\ell^{-1}(\sigma_\ell(r_\ell))}y_\ell)\\
	&= (d^0_{d^0_{\theta_{\ell +1}(r_{\ell +1})}\theta_{\ell +1}(r_\ell)}d^0_{\theta_{\ell +1}(r_{\ell +1})}\theta_{\ell +1},d^0_{r_\ell}y_\ell)\\
	&= (d^0_{d^0_{\theta_{\ell +1}(r_{\ell +1})}\theta_{\ell +1}(r_\ell+1-1)}d^0_{\theta_{\ell +1}(r_{\ell +1})}\theta_{\ell +1},y_{\ell-1})\\
	&= (d^0_{d^0_{\theta_{\ell +1}(r_\ell +1)}\theta_{\ell +1}(r_{\ell+1})}d^0_{\theta_{\ell +1}(r_\ell +1)}\theta_{\ell +1},y_{\ell-1})\\
	&= (d^0_{d^0_{i_{\ell +1}}\theta_{\ell +1}(r_{\ell +1})}d^0_{i_{\ell +1}}\theta_{\ell +1},y_{\ell-1})\\
	&= (d^0_{\theta_\ell(r_{\ell +1})}\theta_\ell,y_{\ell-1})
	= (d^0_{i_\ell}\theta_\ell,y_{\ell-1})
	= (\theta_{\ell-1},y_{\ell-1})
	= (\sigma_{\ell-1},y_{\ell-1}).
	\end{align*}

		Suppose now that \(\phi(\pi)\) and \(\rho\) satisfy condition \ref{adjacent}(iii). Then \(k_\ell = q_{\ell +1} = 0\), \(k_{\ell +1} = q_\ell = 1\), and \(\theta_{\ell}^{-1}(i_\ell) = r_{\ell+1}  < \theta_{\ell}^{-1}(i_{\ell+1}) = r_\ell+1\). Set \(\sigma_\ell = d^0_{\theta_{\ell +1}(r_{\ell +1})}\theta_{\ell +1}\). As before, \(\sigma_\ell\) is a well-defined element of \(S_{\ell - 2\sum \limits_{p=1}^{\ell}q_p}\). Since \(r_\ell \geq r_{\ell +1} = \theta_\ell^{-1}(i_\ell)\), we have
		\begin{align*}
		s_\ell &= \theta_{\ell-1}(r_\ell) = d^0_{i_\ell}\theta_\ell(r_\ell)
		= \left\{\def\arraystretch{1.4}\begin{array}{ll}
		\theta_\ell(r_\ell +1) = i_{\ell +1}, & i_{\ell +1} < i_\ell,\\
		\theta_\ell(r_\ell +1) -1  = i_{\ell +1} -1, & i_{\ell +1} > i_\ell.
		\end{array} \right. 
		\end{align*}
		Since \(r_{\ell +1} < \theta_\ell^{-1}(i_{\ell +1})\), we have 
		\begin{align*}
		s_{\ell+1} &= \theta_{\ell+1}(r_{\ell+1}) = d^1_{i_{\ell+1}}\theta_\ell(r_{\ell+1})
		= \left\{\def\arraystretch{1.4}\begin{array}{ll}
		\theta_\ell(r_{\ell +1}) = i_{\ell}, & i_{\ell} < i_{\ell+1},\\
		\theta_\ell(r_{\ell +1}) -1  = i_{\ell} -1, & i_{\ell} > i_{\ell+1}.
		\end{array} \right. 
		\end{align*}
		Thus either \(s_\ell = i_{\ell +1} < i_\ell =  s_{\ell +1} +1\) or \(i_\ell = s_{\ell +1} < i_{\ell +1} = s_\ell +1\). In the first situation, 
		\begin{align*}
		d^1_{s_\ell}\sigma_{\ell -1} &= d^1_{i_{\ell+1}}\theta_{\ell -1} = d^1_{i_{\ell+1}}d^0_{i_\ell}\theta_{\ell} = d^0_{i_\ell-1}d^1_{i_{\ell+1}}\theta_{\ell}\\ &= d^0_{s_{\ell+1}}\theta_{\ell+1} = d^0_{\theta_{\ell+1}(r_{\ell +1})}\theta_{\ell+1} = \sigma_\ell.
		\end{align*}
		In the second situation, 
		\begin{align*}
		d^1_{s_\ell}\sigma_{\ell -1} &= d^1_{i_{\ell+1}-1}\theta_{\ell -1} = d^1_{i_{\ell+1}-1}d^0_{i_\ell}\theta_{\ell} = d^0_{i_\ell}d^1_{i_{\ell+1}}\theta_{\ell}\\ &= d^0_{s_{\ell+1}}\theta_{\ell+1} = d^0_{\theta_{\ell+1}(r_{\ell +1})}\theta_{\ell+1} = \sigma_\ell.
		\end{align*}
		Hence \(d^1_{s_\ell}(\sigma_{\ell -1},y_{\ell-1}) = (d^1_{s_\ell}\sigma_{\ell -1},d^1_{\sigma_{\ell -1}^{-1}(s_\ell)}y_{\ell-1}) = (\sigma_{\ell},d^1_{r_\ell}y_{\ell-1}) = (\sigma_{\ell},y_{\ell})\) and \(d^0_{s_{\ell +1}}(\sigma_{\ell +1},y_{\ell +1}) = d^0_{\theta_{\ell +1}(r_{\ell +1})}(\theta_{\ell +1},y_{\ell +1}) = (d^0_{\theta_{\ell +1}(r_{\ell +1})}\theta_{\ell +1},d^0_{r_{\ell +1}}y_{\ell +1}) = (\sigma_\ell, y_\ell)\).

	Suppose finally that \(\phi(\pi)\) and \(\rho\) satisfy condition \ref{adjacent}(iv). Then \(k_\ell = q_{\ell +1} = 1\), \(k_{\ell +1} = q_\ell = 0\), and \(r_\ell = \theta_{\ell+1}^{-1}(i_{\ell+1}) < r_{\ell+1} = \theta_{\ell-1}^{-1}(i_\ell)+1\). We also assume that \(i_\ell \leq i_{\ell +1}\) and leave the analogous case \(i_\ell > i_{\ell +1}\) to the reader. Since \(d^1_{i_\ell}\theta_{\ell -1} = \theta_\ell = d^0_{i_{\ell +1}}\theta_{\ell +1}\) and \(\theta_{\ell -1}^{-1}(i_\ell) \geq \theta_{\ell +1}^{-1}(i_{\ell +1})\), \(\mbox{{Proposition \ref{existstheta}}}\) implies that there exists a permutation \(\sigma_\ell \in S_{\ell - 2\sum \limits_{p=1}^{\ell-1} k_p} = S_{\ell - 2\sum \limits_{p=1}^\ell q_p}\) such that \(d^1_{i_\ell} \sigma_\ell = \theta_{\ell +1}\), \(d^0_{i_{\ell +1}+1}\sigma_\ell = \theta_{\ell-1}\), and \(\sigma_\ell ^{-1}(i_\ell) > \sigma_\ell ^{-1}(i_{\ell +1}+1)\). We have 
	\begin{align*}
	r_\ell &= \theta_{\ell +1}^{-1}(i_{\ell+1}) = (d^1_{i_\ell}\sigma_\ell)^{-1}(i_{\ell+1}) = d^1_{\sigma_\ell ^{-1}(i_\ell)}\sigma_\ell ^{-1}(i_{\ell+1}) = \sigma_\ell^{-1}(i_{\ell+1} +1)
	\end{align*}
	and 
	\begin{align*}
	r_{\ell +1} &= \theta_{\ell -1}^{-1}(i_{\ell}) +1 = (d^0_{i_{\ell+1}+1}\sigma_\ell)^{-1}(i_\ell) +1\\ &= d^0_{\sigma_\ell ^{-1}(i_{\ell+1}+1)}\sigma_\ell ^{-1}(i_{\ell}) +1 = \sigma_\ell^{-1}(i_\ell).
	\end{align*}
	Hence \(s_\ell = \sigma_{\ell}(r_\ell) = i_{\ell +1}+1\) and \(s_{\ell +1} = \sigma_{\ell}(r_{\ell +1}) = i_\ell\). 
	We therefore have   
	\(d^0_{s_\ell}(\sigma_\ell, y_\ell) = d^0_{i_{\ell+1}+1}(\sigma_\ell, y_\ell)  = (d^0_{i_{\ell+1}+1}\sigma_\ell, d^0_{\sigma_\ell^{-1}(i_{\ell +1}+1)}y_\ell) = (\theta_{\ell -1}, d^0_{r_\ell}y_\ell) = (\sigma_{\ell -1},y_{\ell -1})\)
	and 
	\(d^1_{s_{\ell+1}}(\sigma_\ell, y_\ell) = d^1_{i_{\ell}}(\sigma_\ell, y_\ell)  = (d^1_{i_{\ell}}\sigma_\ell, d^1_{\sigma_\ell^{-1}(i_{\ell })}y_\ell) = (\theta_{\ell +1}, d^1_{r_{\ell+1}}y_\ell) = (\sigma_{\ell +1},y_{\ell +1})\).
\end{proof}
\end{sloppypar}

\section{Hereditary history-preserving bisimilarity}

A \emph{hereditary history-pre\-serving bisimulation} between two HDAs is a relation \(R\) between their cube paths such that following conditions hold (see  \cite{vanGlabbeek}): 

\begin{enumerate}[\hspace{0.3cm} (1)]
	\item The cube paths of length \(0\) are related.
	\item If \(\pi R \rho\), then \(\mathit{split\mbox{-}trace}(\pi) = \mathit{split\mbox{-}trace}(\rho)\).
	\item If \(\pi R \rho\) and \(\pi \xleftrightarrow{\ell} \pi'\), then \(\exists \rho'\) with \(\rho \xleftrightarrow{\ell} \rho '\) and \(\pi' R \rho'\).
	\item If \(\pi R \rho\) and \(\rho \xleftrightarrow{\ell} \rho '\), then \(\exists \pi'\) with \(\pi \xleftrightarrow{\ell} \pi'\) and \(\pi' R \rho'\).
	\item If \(\pi R \rho\) and \(\pi \to \pi'\), then \(\exists \rho'\) with \(\rho \to \rho '\) and \(\pi' R \rho'\).
	\item If \(\pi R \rho\) and \(\rho \to \rho '\), then \(\exists \pi'\) with \(\pi \to \pi'\) and \(\pi' R \rho'\).
	\item If \(\pi R \rho\), then \(\mathit{end}(\pi)\) is a final state if and only if \(\mathit{end}(\rho)\) is a final state.
	\item If \(\pi R \rho\) and \(\pi' \to \pi\), then \(\exists \rho'\) with \(\rho' \to \rho \) and \(\pi' R \rho'\).
	\item If \(\pi R \rho\) and \(\rho' \to \rho\), then \(\exists \pi'\) with \(\pi' \to \pi\) and \(\pi' R \rho'\).
\end{enumerate}
Two HDAs are called \emph{hhp-bisimilar} if there exists a hereditary history-pre\-serving bisimulation between them.

\begin{theor} \label{symhhp}
	Let \(\Q\) be an HDA. Then \(\Q\) and \(S\Q\) are hhp-bisimilar.
\end{theor}

\begin{proof}
	Consider the relation \(R\) on cube paths of \(S\Q\) and \(\Q\) defined by \[\pi R \rho \Leftrightarrow \rho = \phi(\pi),\]
	where \(\phi(\pi)\) is the cube path defined in Section \ref{seccubepaths}. We show that \(R\) is a hereditary history-preserving bisimulation. Properties (1), (5), (7), (8), and (9) are obvious. Property (3) follows from Proposition \ref{hhp3}. Property (4) follows from Proposition \ref{hhp4lem}. It remains to establish properties (2) and (6).	
	
	(2) By Proposition \ref{edgetheta}, we have 
	\begin{align*}
	\mathit{split\mbox{-}trace}(\pi) &= ((\lambda_{S\Q}(e_{i_j}(\theta_{j-k_j},x_{j-k_j}),k_j))_{j = 1, \dots, m}\\
	&= ((\lambda_{S\Q}(id,e_{\theta_{j-k_j}^{-1}(i_j)}x_{j-k_j}),k_j))_{j = 1, \dots, m}\\
	&= ((\lambda_{\Q}(e_{\theta_{j-k_j}^{-1}(i_j)}x_{j-k_j}),k_j))_{j = 1, \dots, m}\\
	&= \mathit{split\mbox{-}trace}(\phi(\pi)).
	\end{align*}	
	
	(6) Consider a cube path \(\pi\) in \(S\Q\), and suppose that \(\phi(\pi) \to \rho'\). We may suppose that \(\rho' = \phi(\pi)  \tfrac{d^{k}_{r}}{} y\). Let \(\mathit{end}(\pi) = (\theta,x)\), and suppose that \((\theta,x) \in (P_{S\Q})_n = S_n\times (P_\Q)_n\). Then \(\mathit{end}(\phi(\pi)) = x \in (P_\Q)_n\). Suppose first that \(k=1\). Then \(y \in (P_\Q)_{n-1}\), \(r \in \{1, \dots, n\}\), and \(y = d^1_rx\).
	Since 
	\[d^1_{\theta(r)}(\theta,x) = (d^1_{\theta(r)}\theta,d^1_{\theta^{-1}(\theta(r))}x) = (d^1_{\theta(r)}\theta,y),\]
	we may extend \(\pi\) to \(\pi' = \pi \tfrac{d^1_{\theta(r)}}{} (d^1_{\theta(r)}\theta,y)\). We then have \(\phi(\pi') = \phi(\pi) \tfrac{d^1_r}{} y = \rho'\), i.e., \(\pi' R \rho'\).	
	
	Suppose now that \(k = 0\). Then \(y \in (P_\Q)_{n+1}\), \(r \in \{1, \dots, n+1\}\), and \(x = d^0_ry\). Define \(\sigma \in S_{n+1}\) by 
	\[\sigma = (\theta(1)^{\uparrow 1} \;\; \theta(2)^{\uparrow 1} \;\; \cdots \;\; \theta(r-1)^{\uparrow 1}\;\; 1 \;\; \theta(r)^{\uparrow 1} \;\; \cdots\;\; \theta(n)^{\uparrow 1}).\]
	Then \(\sigma^{-1}(1) = r\) and 
	\[d^0_1\sigma = (\theta(1)^{\uparrow 1 \downarrow 1} \;\; \theta(2)^{\uparrow 1\downarrow 1} \;\; \cdots \;\; \theta(r-1)^{\uparrow 1 \downarrow 1}\;\; \theta(r)^{\uparrow 1\downarrow 1} \;\; \cdots\;\; \theta(n)^{\uparrow 1\downarrow 1}) = \theta.\]
	Hence
	\[d^0_1(\sigma,y) = (d^0_1\sigma, d^0_{\sigma^{-1}(1)}y) = (\theta, d^0_ry) = (\theta,x).\]
	We may therefore extend \(\pi\) to \(\pi' = \pi \tfrac{d^0_1}{} (\sigma,y)\). Since \(\sigma^{-1}(1) = r\), we have \(\phi(\pi') = \phi(\pi) \tfrac{d^0_r}{} x = \rho'\), i.e., \(\pi' R \rho'\). 
\end{proof}

If one views symmetric HDAs as HDAs of a particular type, it is natural to define two symmetric HDAs to be hhp-bisimilar if they are hhp-bisimilar as HDAs. From this point of view, symmetric HDAs are a priori at most as expressive as ordinary HDAs. Since, by Theorem \ref{symhhp}, every HDA is hhp-bisimilar to a symmetric one, symmetric HDAs are actually as expressive as ordinary HDAs. The fact that symmetric HDAs are at least as expressive as ordinary HDAs can also be inferred from the following corollary of Theorem \ref{symhhp}:

\begin{cor}
	Two HDAs \(\Q\) and \(\Q'\) are hhp-bisimilar if and only if \(S\Q\) and \(S\Q'\) are hhp-bisimilar. 
\end{cor}

\input{refs.bbl}

\end{document}

%% file: refs.bbl
\providecommand{\bysame}{\leavevmode\hbox to3em{\hrulefill}\thinspace}
\providecommand{\MR}{\relax\ifhmode\unskip\space\fi MR }
\providecommand{\MRhref}[2]{%
  \href{http://www.ams.org/mathscinet-getitem?mr=#1}{#2}
}
\providecommand{\href}[2]{#2}

%% file: shda.bbl
\begin{thebibliography}{Gou02}

\bibitem[FL91]{FiedorowiczLoday}
Z.~Fiedorowicz and J.-L. Loday, \emph{{Crossed simplicial groups and their
  associated homology}}, Transactions of the American Mathematical Society
  \textbf{326} (1991), no.~1, 57--87.

\bibitem[Gau10]{GaucherCombinatorics}
P.~Gaucher, \emph{{Combinatorics of labelling in higher-dimensional automata}},
  Theoretical Computer Science \textbf{411} (2010), 1452--1483.

\bibitem[GJ99]{GoerssJardine}
{P.} Goerss and {J.} Jardine, \emph{{Simplicial Homotopy Theory}}, Progress in
  Mathematics, vol. 174, Birkhäuser Verlag, 1999.

\bibitem[Gla06]{vanGlabbeek}
R.J.~van Glabbeek, \emph{{On the expressiveness of higher dimensional
  automata}}, Theoretical Computer Science \textbf{356} (2006), no.~3,
  265--290.

\bibitem[GM12]{GoubaultMimram}
E.~Goubault and S.~Mimram, \emph{{Formal relationships between geometrical and
  classical models for concurrency}}, Electronic Notes in Theoretical Computer
  Science \textbf{283} (2012), 77--109.

\bibitem[Gou02]{GoubaultLabCubATS}
E.~Goubault, \emph{{Labelled cubical sets and asynchronous transition systems:
  an adjunction}}, {CMCIM'02}, 2002, pp.~1--30.

\bibitem[Kra87]{Krasauskas}
R.~Krasauskas, \emph{{Skew-simplicial groups}}, Lithuanian Mathematical Journal
  \textbf{27} (1987), 47--54.

\bibitem[Pra91]{Pratt}
V.~Pratt, \emph{{Modeling Concurrency with Geometry}}, {POPL '91, Proceedings
  of the 18th ACM SIGPLAN-SIGACT symposium on Principles of programming
  languages}, ACM New York, NY, USA, 1991, pp.~311--322.

\end{thebibliography}
